\documentclass[letterpaper, 11pt]{article}
\usepackage{appendix}
\usepackage{array} 
 
\newcommand{\comment}[1]{}
 
\usepackage{graphicx}
\usepackage{subfigure}
\usepackage{amsmath}
\usepackage{mdwlist}
\usepackage{pxfonts}
 
\usepackage{algorithm}

\usepackage{hyperref}

\usepackage{xspace}

\newenvironment{proof}{\paragraph{\bf Proof:}}{\hspace*{\fill}\(\Box\)}
\newenvironment{proofSketch}{\paragraph{\bf Proof Sketch:}}{\hspace*{\fill}\(\Box\)}
\newtheorem{theorem}{Theorem}

\newtheorem{definition}{Definition}

\newtheorem{lemma}{Lemma}

\def\noflash#1{\setbox0=\hbox{#1}\hbox to 1\wd0{\hfill}}

\newcommand{\scriptf}{\mathcal{F}}
\newcommand{\scripte}{\mathcal{E}}
\newcommand{\scriptv}{\mathcal{V}}
\newcommand{\scriptl}{\mathcal{L}}
\newcommand{\scripts}{\mathcal{S}}
\newcommand{\matrixm}{\textbf{M}}
\newcommand{\matrixh}{\textbf{H}}
\newcommand{\graphh}{\textit{H}}

\begin{document}
\title{Iterative Approximate Byzantine Consensus\\ under a {\em Generalized} Fault Model\footnote{\normalsize This research is supported in part by National Science Foundation award CNS 1059540
and Army Research Office grant W-911-NF-0710287. Any opinions, findings,
and conclusions or recommendations expressed here are those of the authors and do not
necessarily reflect the views of the funding agencies or the U.S. government.}}

\author{Lewis Tseng$^{1,3}$, and Nitin Vaidya$^{2,3}$\\~\\
 \normalsize $^1$ Department of Computer Science,\\
 \normalsize $^2$ Department of Electrical and Computer Engineering, 
 and\\ \normalsize $^3$ Coordinated Science Laboratory\\ \normalsize University of Illinois at Urbana-Champaign\\ \normalsize Email: \{ltseng3, nhv\}@illinois.edu~\\~\\Technical Report}

\date{May 21, 2012}
\maketitle

\newcommand{\bfA}{{\bf A}}
\newcommand{\bfH}{{\bf H}}
\newcommand{\bfQ}{{\bf Q}}
\newcommand{\bfM}{{\bf M}}
\newcommand{\bfv}{v}

\begin{abstract}
In this work, we consider
a {\em generalized} fault model that can be used to represent a wide range of failure scenarios, including correlated failures and non-uniform node reliabilities. This fault model is general in the sense that fault models studied in prior related work, such as $f$-total and $f$-local models, are special cases of the generalized fault model. Under the generalized fault model, we explore iterative approximate Byzantine consensus (IABC) algorithms in arbitrary directed networks. We prove a necessary and sufficient condition for the existence of IABC algorithms.
The use of the generalized fault model helps to gain a better understanding 
of IABC algorithms.

\end{abstract}

\newpage

\section{Introduction}
\label{sec:intro}

Dolev et al. \cite{AA_Dolev_1986} introduced the notion of
{\em approximate Byzantine consensus} by relaxing the requirement
of {\em exact} consensus \cite{AA_nancy}.
The goal in approximate consensus is to allow the fault-free nodes to
agree on values that are approximately equal to each other (and {not necessarily}
exactly identical). 
In presence of Byzantine faults, while {\em exact} consensus 
is impossible in {\em asynchronous} systems \cite{FLP_one_crash}, approximate
consensus is achievable \cite{AA_Dolev_1986}.
The notion of approximate consensus is of interest in {\em synchronous}
systems as well, since approximate consensus can be achieved using
distributed algorithms that do {not} require complete knowledge of
the network topology \cite{AA_convergence_markov}.
The rest of the discussion in this paper assumes 
a {\em synchronous} systems.

The fault model assumed in much of the work on Byzantine consensus allows up to $f$
Byzantine faulty nodes in the network. We will refer to this fault model as the ``$f$-total''
fault model \cite{IBA_broadcast_Sundaram, psl_BG_1982, AA_Dolev_1986, AA_nancy}. In prior work, other fault models have been explored as well.
For instance, in the ``$f$-local'' fault model, up to $f$ neighbors of {\em each}
node in the network may be faulty \cite{Koo_radio_byzantine, Vartika_radio_byzantine_2005, IBA_broadcast_Sundaram}, and in the
$f$-fraction model \cite{IBA_broadcast_Sundaram}, up to $f$ fraction of the neighbors of each node may be faulty.
In this paper, we consider a {\em generalized} fault model (to be described in the next
section).
The {\em generalized fault} model specifies a ``fault domain'', which is
a collection of feasible fault sets (a similar fault model is recently presented in \cite{Nonuniform_failure_models}). For example, in a system consisting of four nodes, namely,
nodes $1, 2, 3$ and $4$, the fault domain could be specified as $\scriptf = \{\,\{1\},\, \{2,3,4\}\,\}$.
Thus, in this case, either node 1 may be faulty, or any subset of nodes in $\{2,3,4\}$ may be faulty. However,
node 1 may not be faulty simultaneously with another node.
The new fault model is general in the sense that the other fault models studied in the literature,
such as $f$-total, $f$-local and $f$-fraction models, are special cases of the generalized fault model.

Analysis of consensus under the generalized fault model offers
some new insights into how the choice of the fault model affects
algorithm design. 
In particular,
we consider ``iterative'' algorithms
for achieving approximate Byzantine consensus in synchronous point-to-point
networks that are modeled by arbitrary {\em directed}\, graphs. The {\em iterative
approximate Byzantine consensus} (IABC) algorithms of interest have
the following properties, which we will soon state more formally:
\begin{itemize}
\item {\em Initial state} of each node is equal to a real-valued
{\em input} provided to that node.
\item {\em Validity} condition: After each iteration of an IABC algorithm, the state of each fault-free node
must remain in the {\em convex hull} of the states of the fault-free nodes
at the end of the {\em previous} iteration.
\item {\em Convergence} condition:
For any $\epsilon>0$, after a sufficiently large number of iterations,
the states of the fault-free nodes are guaranteed to be within $\epsilon$
of each other.
\end{itemize}

This paper is a generalization of our recent work on
IABC algorithms under the $f$-total fault
model \cite{vaidya_PODC12,vaidya_matrix_IABC}. 
The contributions of this paper are as follows:
\begin{itemize}
\item We identify a necessary condition on the communication graph for the
existence of a correct IABC algorithm under the {\em generalized} fault model
(Sections \ref{sec:iabc} and \ref{s_necessity}).
\item We introduce a new IABC algorithm for the generalized fault model
(Section \ref{s_algo}) that uses  only ``local'' information.
\item A transition matrix representation of the new IABC algorithm is
presented (Section \ref{s_sufficiency}). This representation is then used
to prove the correctness of the proposed algorithm (Section \ref{s_proof}).
\end{itemize}
Since the results here generalize our prior results \cite{vaidya_PODC12,vaidya_matrix_IABC}, naturally the proof techniques used here have some similarities
to the prior work. The material in Section \ref{s_proof} bears the
strongest similarity to our prior work.
The rest of the paper, however, presents results that provide new intuition
on the problem of approximate consensus. In particular, materials in Sections \ref{s_necessity} and \ref{s_algo} shed light on how the fault model influences the design of IABC algorithms.

\section{Models}
\label{s_models}

\paragraph{Communication Model:}

The system is assumed to be {\em synchronous}.
The communication network is modeled as a simple {\em directed} graph $G(\scriptv,\scripte)$, where $\scriptv=\{1,\dots,n\}$ is the set of $n$ nodes, and $\scripte$ is the set of directed edges between the nodes in $\scriptv$.
 We assume that $n\geq 2$, since the consensus problem for $n=1$ is trivial.
 Node $i$ can reliably transmit messages to node $j$ if and only if
the directed edge $(i,j)$ is in $\scripte$.
Each node can send messages to itself as well, however,
for convenience, we {exclude self-loops} from set $\scripte$.
That is, $(i,i)\not\in\scripte$ for $i\in\scriptv$.
With a slight abuse of terminology, we will use the terms {\em edge}
and {\em link} interchangeably in our presentation.

For each node $i$, let $N_i^-$ be the set of nodes from which $i$ has incoming
edges.
That is, $N_i^- = \{\, j ~|~ (j,i)\in \scripte\, \}$.
Similarly, define $N_i^+$ as the set of nodes to which node $i$
has outgoing edges. That is, $N_i^+ = \{\, j ~|~ (i,j)\in \scripte\, \}$. Nodes in $N_i^-$ and $N_i^+$ are, respectively, said to be incoming and outgoing neighbors of node $i$.
Since we exclude self-loops from $\scripte$,
$i\not\in N_i^-$ and $i\not\in N_i^+$. 
However, we note again that each node can indeed send messages to itself. 

\paragraph{Generalized Byzantine Failure Model:}

We consider the Byzantine failure model, with possible faulty nodes specified using a ``fault domain'' $\scriptf$ (defined below). A faulty node may {\em misbehave} arbitrarily. Possible misbehavior includes transmitting incorrect and mismatching (or inconsistent) messages to different neighbors. The faulty nodes may collaborate with each other.  Moreover, the faulty nodes are assumed to have a complete knowledge of the execution of the algorithm, including the states of all the nodes, the algorithm specification, and the network topology.

The generalized fault model is characterized using {\em fault domain}
$\scriptf \subseteq 2^{\scriptv}$ as follows:
Nodes in set $F$
may fail during an execution of the algorithm only if there exists set $F^*\in \scriptf$
such that $F\subseteq F^*$.
Set $F$ is then said to be a {\em feasible} fault set.
\begin{definition}
Set $F\subseteq \scriptv$ is said to be a \,\underline{feasible}\, fault set, if there
exists $F^*\in\scriptf$ such that $F\subseteq F^*$.
\end{definition}
Thus, each set in $\scriptf$ specifies nodes that may
all potentially fail during a single execution of the algorithm (a similar fault model is also considered in \cite{Nonuniform_failure_models}). This feature
can be used to capture the notion of correlated failures.
 For example, consider a system consisting
of four nodes, namely, nodes 1, 2, 3, and 4. Suppose that
\[ \scriptf = \{ \, \{1\}, \{2\}, \{3,4\}\, \}\]
This definition of $\scriptf$ implies that during an execution either (i) node 1 may
fail, or (ii) node 2 may fail, or (iii) any subset of $\{3,4\}$ may fail, and no
other combination of nodes may fail (e.g., nodes 1 and 3 cannot both fail in a single execution). In this case, the reason that the set
$\{3, 4\}$ is in the fault domain may be that the failures of nodes
3 and 4 are correlated.

The generalized fault model is also useful to capture variations in node reliability. For instance, in the above example,
nodes 1 and 2 may be more reliable than nodes 3 and 4. Therefore, while
simultaneous failure of nodes 3 and 4 may occur, simultaneous failure of nodes
1 and 2 is less likely. Therefore, $\{1,2\}\not\in\scriptf$.

\noindent {\em Local knowledge of $\scriptf$:}
To implement our IABC Algorithm presented in Section \ref{s_algo}, it is
sufficient for 
each node $i$ to know $N_i^-\cap F$, for each feasible fault set $F$.
In other words, each node only needs to know the set of its incoming neighbors that may fail simultaneously. Thus, the iterative algorithm can be implemented using only ``local'' information  regarding $\scriptf$.

\section{Iterative Approximate Byzantine Consensus (IABC) Algorithms}
\label{sec:iabc}

In this section, we describe the structure of the IABC algorithms of interest,
and state the validity and convergence conditions that they must satisfy.

Each node $i$ maintains state $v_i$, with $v_i[t]$ denoting the state
of node $i$ at the {\em end}\, of the $t$-th iteration of the algorithm.
Initial state of node $i$,
$v_i[0]$, is equal to the initial {\em input}\, provided to node $i$.
At the {\em start} of the $t$-th iteration ($t>0$), the state of
node $i$ is $v_i[t-1]$.
The IABC algorithms of interest will require each node $i$
to perform the following three steps in iteration $t$ where $t>0$.
Note that the faulty nodes may deviate from this specification. 

\begin{enumerate}
\item {\em Transmit step:} Transmit current state, namely $v_i[t-1]$, on all outgoing edges and self-loop
 (to nodes in $N_i^+$ and node $i$ itself).

\item {\em Receive step:} Receive values on all incoming edges and self-loop (from nodes in $N_i^-$ and itself). 
Denote by $r_i[t]$ the vector of values received by node $i$ from its
incoming neighbors and itself. The size of vector $r_i[t]$ is $|N_i^-|+1$.

\item {\em Update step:} Node $i$ updates its state using a transition function $Z_i$ as
follows. $Z_i$ is a part of the specification of the algorithm, and takes
 the vector $r_i[t]$ as the input.
\begin{eqnarray}
v_i[t] & = &  Z_i ~( ~r_i[t]~)
\label{eq:Z_i}
\end{eqnarray}


\end{enumerate}
\comment{++++++++++++++++
We now define $U[t]$ and $\mu[t]$, assuming that $F \in \scriptf$
is the set of Byzantine faulty nodes, with the nodes
in $\scriptv-F$ being fault-free.\footnote{\normalsize For sets $X$ and $Y$, $X-Y$ contains elements that are in $X$ but not in $Y$. That is, $X-Y=\{i~|~ i\in X,~i\not\in Y\}$.} 
\begin{itemize}

\item $U[t] = \max_{i\in\scriptv-F}\,v_i[t]$. $U[t]$ is the largest state among the fault-free nodes at the end of the $t$-th iteration.
Since the initial state of each node is equal to its input,
$U[0]$ is equal to the maximum value of the initial input at the fault-free nodes.

\item $\mu[t] = \min_{i\in\scriptv-F}\,v_i[t]$. $\mu[t]$ is the smallest state among the fault-free nodes at the end of the $t$-th iteration.
$\mu[0]$ is equal to the minimum value of the initial input at the
fault-free nodes.
\end{itemize}
++++++++++++++++++}

\noindent
The following conditions must be satisfied by an IABC algorithm when
the set of faulty nodes (in a given execution) is $F$:

\begin{itemize}
\item {\em Validity:} $\forall t>0$, and all fault-free nodes $i\in\scriptv-F$,\\
\hspace*{0.5in} $v_i[t] \geq \min_{j\in\scriptv-F}~ v_j[t-1]$
~\mbox{~~and~~}~
~~ $v_i[t] \leq \max_{j\in\scriptv-F}~ v_j[t-1]$.\footnote{For sets $X$ and $Y$, $X-Y$ contains elements
that are in $X$ but not in $Y$. That is,
$X-Y=\{i~|~ i\in X,~i\not\in Y\}$.
}

\item {\em Convergence:} for all {\em fault-free} nodes $i,j\in \scriptv-F$,~~ $\lim_{\,t\rightarrow\infty} ~ (v_i[t]-v_j[t]) = 0$
\end{itemize}

An IABC algorithm is said to be {\em correct} if it satisfies the above validity and convergence conditions in the given graph $G(\scriptv, \scripte)$. For a given fault domain $\scriptf$
for graph $G(\scriptv,\scripte)$, the objective here is to identify the necessary and sufficient conditions for the existence of a {\em correct} IABC algorithm.

\section{Necessary Condition}
\label{s_necessity}

In this section, we develop a necessary condition for the existence of a correct IABC algorithm. The necessary condition will be proved to be also sufficient in Section  \ref{s_sufficiency}.


\subsection{Preliminaries}

To facilitate the statement of the necessary condition, we first
 introduce the notions of ``source component'' and ``reduced graph'' using the following three definitions. 

\begin{definition}
\label{def:decompose}
{\bf Graph Decomposition:}
Let $H$ be a directed graph. Partition graph $H$ into strongly connected components, $H_1,H_2,\cdots,H_h$, where $h$ is a non-zero integer dependent on graph $H$, such that

\begin{itemize}
\item every pair of nodes {\bf within} the same strongly connected component has directed paths in $H$ to each other, and

\item for each pair of nodes, say $i$ and $j$, that belong to two {\bf different} strongly connected components, either $i$ does not have a directed path to $j$ in $H$, or $j$ does not have a directed path to $i$ in $H$.
\end{itemize}
Construct a graph $H^d$ wherein each strongly connected component $H_k$ above is represented by vertex $c_k$, and there is an edge from vertex $c_k$ to vertex $c_l$ if and only if the nodes in $H_k$ have directed paths in $H$ to the nodes in $H_l$. $H^d$ is called the \underline{decomposition graph} of $H$.
\end{definition}

\noindent
It is known that for any directed graph $H$, the corresponding decomposition graph $H^d$ is a directed {\em acyclic} graph (DAG) \cite{dag_decomposition}.

\begin{definition}
{\bf Source Component}:
Let $H$ be a directed graph, and let $H^d$ be its decomposition graph as per
Definition~\ref{def:decompose}.  Strongly connected component $H_k$ of $H$ is said to be a {\em source component} if the corresponding vertex $c_k$ in $H^d$ is \underline{not} reachable from any other vertex in $H^d$.
\end{definition}


\begin{definition}
\label{def:reduced} {\bf Reduced Graph:}
For a given graph $G(\scriptv,\scripte)$ and a feasible fault set $F$, a \underline{reduced graph} $G_F(\scriptv_F,\scripte_F)$ is obtained as follows:
\begin{itemize}
\item Node set is obtained as $\scriptv_F = \scriptv-F$.
\item For each node $i\in\scriptv_F$, a feasible fault set $F_x(i)$ is chosen,
	and then the edge set $\scripte_F$ is obtained as follows:
\begin{itemize} 
\item remove from $\scripte$ all the links incident on the nodes in $F$, and
\item for each $i \in \scriptv_F$ and each $j\in F_x(i)\cap \scriptv_F \cap N_i^-$,
remove link $(j, i)$ from $\scripte$.
\end{itemize}
Feasible fault sets $F_x(i)$ and $F_x(j)$ chosen for $i\neq j$ may or may not be identical.
\end{itemize}
\end{definition}

Note that for a given $G(\scriptv, \scripte)$ and a given $F$, multiple reduced graphs $G_F$ may exist, depending on the choice of $F_x$ sets above.


\subsection{Necessary Condition}
\label{ss_necessity}

For a correct IABC algorithm to exist, the network graph $G(\scriptv, \scripte)$ must satisfy the necessary condition stated in Theorem \ref{thm:nc} below.

\begin{theorem}
\label{thm:nc}
Suppose that a correct IABC algorithm exists for $G(\scriptv, \scripte)$. Then, any
reduced graph $G_F$, corresponding to any feasible fault set $F$,
must contain exactly one {\em source component}.
\end{theorem}

\begin{proofSketch}
A complete proof is presented in Appendix \ref{app_s_necessity}.
The proof is by contradiction. Let us assume that a correct IABC algorithm exists, and for some feasible fault set $F$, and  feasible sets $F_x(i)$ for each $i\in\scriptv-F$, the resulting reduced graph contains two source components. Let $L$ and $R$ denote the nodes in the two source components, respectively. Thus, $L$ and $R$ are disjoint and non-empty. Let $C = (\scriptv - F - L -R)$ be the remaining nodes in the reduced graph. $C$ may or may not be non-empty. Assume that the nodes in $F$ (if non-empty) are all faulty, and all the nodes in $L$, $R$, and $C$ (if non-empty) are fault-free. 
Suppose that each node in $L$ has initial input equal to $m$, each node in $R$
has initial input equal to $M$, where $M>m$, and each node in $C$ has
an input in the range $[m,M]$. As elaborated in Appendix \ref{app_s_necessity},
the faulty nodes can behave in such a manner that, in each iteration, nodes
in $L$ and $R$ are forced to maintain their updated state equal to
$m$ and $M$, respectively, 
so as to satisfy the {\em validity} condition. This ensures that,
no matter how many iterations are performed, the {\em convergence}
condition cannot be satisfied.
\end{proofSketch}

\section{Algorithm 1}
\label{s_algo}

We will prove that there exists an IABC algorithm -- particularly {\em Algorithm 1} below -- that satisfies the {\em validity} and {\em convergence} conditions provided that the graph $G(\scriptv,\scripte)$ satisfies the necessary condition in Theorem~\ref{thm:nc}.
This implies that the necessary condition in Theorem~\ref{thm:nc} is also sufficient. {\em Algorithm 1} has the three-step structure described in Section \ref{sec:iabc}. This algorithm is
a generalization -- to accommodate the generalized fault model -- of iterative algorithms that
 were analyzed in prior work \cite{AA_Dolev_1986,AA_nancy,AA_PCN_Local, leblanc_HiCoNs}, including in our own prior work as well \cite{vaidya_PODC12, vaidya_matrix_IABC}.
The key difference from previous algorithms is in the {\em Update} step below.

~

\newpage
\hrule
{\bf Algorithm 1}
\vspace*{4pt}\hrule

\begin{enumerate}

\item {\em Transmit step:} Transmit current state $v_i[t-1]$ on all outgoing edges and self-loop.
\item {Receive step:} Receive values on all incoming edges and self-loop. These values form vector $r_i[t]$ of size $|N_i^-|+1$ (including the value from node $i$ itself). When a fault-free node expects to receive a message from an incoming neighbor but does not receive the message, the message value is assumed to be equal to some {\em default value}.

\item {\em Update step:}
Sort the values in $r_i[t]$ in an increasing order (breaking ties arbitrarily). Let $D$ be a vector of nodes arranged in an order ``consistent'' with $r_i[t]$: specifically, $D(1)$  is the node that sent the smallest value in $r_i[t]$, $D(2)$ is the node that sent the second smallest value in $r_i[t]$, and so on. The size of vector $D$ is also $|N_i^-|+1$.

From vector $r_i[t]$, eliminate the smallest $f_1$ values, and the largest $f_2$ values, where $f_1$ and $f_2$ are defined as follows:
\begin{itemize}
\item $f_1$ is the largest number such that there exists
a feasible fault set $F' \subseteq N_i^-$ containing nodes $D(1), D(2), ..., D(f_1)$. Recall that $i \not\in N_i^-$.
\item $f_2$ is the largest number such that there exists a feasible fault set
$F'' \subseteq N_i^-$ containing nodes $D(|N_i^-|-f_2+2), D(|N_i^-|-f_2+3), ..., D(|N_i^-|+1)$.
\end{itemize} 
$F'$ and $F''$ above may or may not be identical. 

Let $N_i^*[t]$ denote the set of nodes from whom the remaining $|N_i^-| +1 - f_1 - f_2$ values in $r_i[t]$ were received, and let $w_j$ denote the value received from node $j\in N_i^*[t]$. Note that $i \in N_i^*[t]$. Hence, for convenience, define $w_i=v_i[t-1]$ to be the value node $i$ ``receives'' from itself.  Observe that
if $j\in N_i^*[t]$ is fault-free, then $w_j=v_j[t-1]$.

Define
\begin{eqnarray}
v_i[t] ~ = ~ Z_i(r_i[t]) ~ = ~\sum_{j\in N_i^*[t]} a_i \, w_j
\label{e_Z}
\end{eqnarray}
where
\[ a_i = \frac{1}{|N_i^*[t]|} = \frac{1}{|N_i^-|+1-f_1 - f_2}
\] 


The ``weight'' of each term on the right-hand side of
(\ref{e_Z}) is $a_i$,
 and these weights add to 1. Also, $0<a_i\leq 1$. 
Although $f_1, f_2$ and $a_i$ may be different for each iteration $t$, for simplicity, we do not explicitly represent this dependence on $t$ in the notations.

\end{enumerate}

\hrule

~

Observe $f_1+f_2$ nodes whose values are eliminated in the {\em Update} step above are all in $N_i^-$. Thus, the above algorithm can be implemented
by node $i$ if it knows which of its incoming neighbors may fail simultaneously; node $i$ does not need to know the entire fault domain $\scriptf$ as such.

The main difference between the above algorithm and IABC algorithms
in prior work is in the choice of the values eliminated
from vector $r_i[t]$ in the {\em Update} step.
 The manner in which the values are eliminated
ensures that the values received from nodes $D(f_1+1)$ and $D(|N_i^-|-f_2+1)$
(i.e., the smallest and largest values that survive in $r_i[t]$) are
within the convex hull of the state of fault-free nodes, even if nodes
$D(f_1+1)$ and $D(|N_i^-|-f_2+1)$ may not be fault-free. 
This property is useful in proving algorithm correctness (as discussed
below).

\section{Sufficiency}
\label{s_sufficiency}

We will show that Algorithm 1 satisfies validity and convergence conditions,
provided that $G(\scriptv, \scripte)$ satisfies the condition below, which matches the necessary condition stated in Theorem \ref{thm:nc}.

\noindent{\bf Sufficient condition:}
{\em 
Any reduced graph $G_F$ corresponding to any feasible fault set $F$ contains
exactly one {\em source component}.
}

In the rest of this section, we assume that $G(\scriptv,\scriptf)$ satisfies
the above condition. To prove its sufficiency, we first develop a {\em transition matrix} representation of the {\em Update} step in Algorithm 1.

\subsection{Transition Matrix Representation}

In our discussion below, $\matrixm[t]$ is a square matrix, $\matrixm_i[t]$ is the $i$-th row
of the matrix, and $\matrixm_{ij}[t]$ is the element at the intersection of the $i$-th
row and $j$-th column of $\matrixm[t]$.

For a given execution of Algorithm 1, let $F$ denote the actual set of faulty nodes in that
execution. Let $|F|=\psi$. Without loss of generality, suppose that 
nodes $1$ through $(n-\psi)$ are fault-free, and if $\psi > 0$, nodes $(n-\psi+1)$ through $n$ are faulty. Denote by $v[0]$ the column vector consisting of the initial states of all the fault-free nodes. Denote by $v[t]$, where $t \geq 1$, the column vector consisting of the states of all the fault-free nodes at the end of the $t$-th iteration. The $i$-th element of vector $v[t]$ is state $v_i[t]$. The size of vector $v[t]$ is $(n - \psi)$. 

We will show that
the iterative update of the state of a fault-free node $i~(1 \leq i \leq n-\psi )$ performed in (\ref{e_Z}) in Algorithm 1 can be expressed using the matrix form below.

\begin{equation}
\label{matrix:e_Z}
v_i[t] = \matrixm_i[t]~v[t-1]
\end{equation}
where $\matrixm_i[t]$ is a {\em stochastic row} vector of size $n - \psi$. That is, $\matrixm_{ij}[t] \geq 0$, for $1 \leq j \leq n-\psi$, and $\sum_{1 \leq j \leq n-\psi} \matrixm_{ij}[t] = 1$.\footnote{\comment{===== old note=====Recall that the operation in the update state is dependent on the values sent
by the incoming neighbors, some of which may be faulty. In other words, new state of
a node depends on the behavior of the faulty neighbors as well.
Therefore, as will be seen later,
in addition to $t$, the row vector $\matrixm_i[t]$ may depend on the state vector $v[t-1]$
as well as the behavior of the faulty nodes in $F$.
For simplicity, the notation $\matrixm_i[t]$ does not explicitly represent this dependence.==========}In addition to $t$, the row vector $\matrixm_i[t]$ may depend on the state vector $v[t-1]$
as well as the behavior of the faulty nodes in $F$.
For simplicity, the notation $\matrixm_i[t]$ does not explicitly represent this dependence.} By ``stacking'' (\ref{matrix:e_Z}) for different $i$, $1 \leq i \leq n-\psi$, we will represent the {\em Update} step of Algorithm 1 at all the fault-free nodes together using (\ref{matrix:alg1}) below.
\begin{equation}
\label{matrix:alg1}
v[t] = \matrixm[t]~v[t-1]
\end{equation}
where $\matrixm[t]$ is a $(n-\psi) \times (n-\psi)$ {\em row stochastic} matrix, with its $i$-th row being equal to $\matrixm_i[t]$ in (\ref{matrix:e_Z}). $\matrixm[t]$ is said to be a
\underline{transition matrix}.

In the rest of this section, we will first ``construct'' a transition matrix
$\matrixm[t]$ that satisfies
certain desirable properties. Then, we will identify a connection between the
transition matrix and the sufficiency condition stated above, and use this connection
to establish {\em convergence} property for Algorithm 1. The {\em validity} property
also follows from the transition matrix representation. 

\subsection{Construction of the Transition Matrix}
\label{s_construction}

We will construct a transition matrix with the property described in Lemma \ref{lemma:tm2cm} below. 

\begin{lemma}
\label{lemma:tm2cm}
The {\em Update step} of Algorithm 1 at the fault-free nodes can be
expressed using row stochastic transition matrix {\normalfont$\matrixm[t]$}, such that
there exists a feasible fault set {\normalfont$F_x(i)$} for each {\normalfont $i\in \scriptv-F$}
 such that,  for all {\normalfont$j \in \{i\} \cup ((\scriptv_{F} - F_x(i)) \cap N_i^-)$},

\[ \text{{\normalfont $\matrixm_{ij}[t] ~ \geq ~ \beta$}}
\]

where $\beta$ is a constant (to be defined later), and $0<\beta\leq 1$.
\end{lemma}

In \cite{vaidya_matrix_IABC} as well, we construct a transition matrix 
to prove correctness of an IABC algorithm under the $f$-total fault model.
However, the {\em generalized} fault model introduces additional complexity,
which is handled here using a new approach to construct the transition
matrix.

\begin{proof}
We prove the correctness of Lemma \ref{lemma:tm2cm} by constructing $\matrixm_i[t]$ for $1 \leq i \leq n - \psi$ that satisfies the conditions in Lemma \ref{lemma:tm2cm}. Recall that
$F$ is the set of faulty nodes, and $|F|=\psi$. As stated before,
without loss of generality, nodes $1$ through $n-\psi$ are assumed to be
fault-free, and the remaining $\psi$ nodes faulty.

Consider a fault-free node $i$ performing the {\em Update} step in Algorithm 1.
In the {\em Update} step, recall that the smallest $f_1$ and the largest $f_2$ values are
eliminated from $r_i[t]$, where the choice of $f_1$ and $f_2$ is described
in Algorithm 1.
 Let us denote by $\scripts$ and $\scriptl$, respectively, the set of nodes\footnote{Although $\scripts$ and $\scriptl$ may be different for each $t$, for simplicity, we do not explicitly represent this dependence on $t$ in the notations $\scripts$ and $\scriptl$.} from whom the smallest $f_1$ and the largest $f_2$ values were received by node $i$ in iteration $t$. Define sets $\scripts_g$ and $\scriptl_g$ to be subsets of $\scripts$ and $\scriptl$ that contain all the fault-free nodes in $\scripts$ and $\scriptl$, respectively. That is, $\scripts_g = \scripts \cap (\scriptv - F)$ and $\scriptl_g = \scriptl \cap (\scriptv - F)$.

Construction of $\matrixm_i[t]$ differs somewhat depending on whether
sets $\scripts_g, \scriptl_g$ and $N_i^*[t] \cap F$ are empty or non-empty.
We divide the possibilities into 6 separate cases.  Due to space limitation,
here we present the construction for one of the cases (named Case I).
The construction for the remaining
cases is presented in Appendix \ref{app_s_construction}.

In Case I,
$\scripts_g \neq \Phi, \scriptl_g \neq \Phi$, and $N_i^*[t] \cap F \neq \Phi$. Let $m_{\scripts}$ and $m_{\scriptl}$ be defined as shown below. Recall that the nodes in $\scripts_g$ and $\scriptl_g$ are all fault-free, and therefore, for any node $j\in \scripts_g \cup\scriptl_g$, $w_j=v_j[t-1]$ (in the notation of Algorithm 1).

\begin{equation*}
m_{\scripts} = \frac{\sum_{j \in \scripts_g} v_j[t-1]}{|\scripts_g|}~~~~~\text{and}~~~~~m_{\scriptl} = \frac{\sum_{j \in \scriptl_g} v_j[t-1]}{|\scriptl_g|}
\end{equation*}
Now, consider any node $k \in N_i^*[t]$. By the definition of sets $\scripts_g$ and $\scriptl_g$,
 $m_{\scripts} \leq w_k \leq m_{\scriptl}$. Therefore, we can find weights $S_k \geq 0$ and $L_k \geq 0$ such that $S_k + L_k = 1$, and

\begin{eqnarray}
w_k & = & S_k~m_{\scripts} + L_k~m_{\scriptl} \\
& = & 
\frac{S_k}{|\scripts_g|}
\sum_{j \in \scripts_g} v_j[t-1]
+
\frac{L_k}{|\scriptl_g|}
\sum_{j \in \scriptl_g} v_j[t-1]
\label{eq:caseI}
\end{eqnarray}
Clearly, at least one of $S_k$ and $L_k$ must be $\geq 1/2$.
%
%
We now define elements $\matrixm_{ij}[t]$ of row $\matrixm_i[t]$:

\begin{itemize}
\item For $j \in N_i^*[t]\cap(\scriptv-F)$ : In this case, $j$ is either a fault-free
incoming neighbor of $i$, or $i$ itself.
For each such $j$, define $\matrixm_{ij}[t] = a_i$. This is obtained by observing
in (\ref{e_Z}) that the contribution of such a node $j$ to the new state
$v_i[t]$ is $a_i~w_j = a_i~v_j[t-1]$.

The elements of $\matrixm_i[t]$ defined here add up to $$|N_i^*[t] \cap (\scriptv - F)|~a_i$$

\item For $j\in \scripts_g\cup\scriptl_g$ : In this case, $j$ is a fault-free node in $\scripts$ or  $\scriptl$.

For each $j \in \scripts_g$,
\[
\matrixm_{ij}[t] ~=~ a_i \, \sum_{k \in N_i^*[t] \cap F} \frac{S_k}{|\scripts_g|}
\]
and for each node $j \in \scriptl_g$,
\[
\matrixm_{ij}[t] ~=~ a_i \, \sum_{k \in N_i^*[t] \cap F} \frac{L_k}{|\scriptl_g|}
\]
To obtain these two expressions, we represent value $w_k$ sent by each faulty node $k$ in $N_i^*[t]$, i.e., $k\in N_i^*[t] \cap F$, using (\ref{eq:caseI}).  
Recall that this node $k$ contributes $a_iw_k$ to (\ref{e_Z}).
The above two expressions are then obtained by summing (\ref{eq:caseI})
over all the faulty nodes in $N_i^*[t]\cap F$, and replacing this sum
by equivalent contributions by nodes in $\scripts_g$ and $\scriptl_g$.

The elements of $\matrixm_i[t]$ defined here add up to $$a_i \, \sum_{k \in N_i^*[t] \cap F} (S_k + L_k) = |N_i^*[t] \cap F|~a_i.$$

\item For $j\in (\scriptv - F) - (N_i^*[t] \cup \scripts_g \cup \scriptl_g)$ :
 These fault-free nodes have not yet been considered above.
For each such node $j$, define $\matrixm_{ij}[t] = 0$.
\end{itemize}
With the above definition of $\matrixm_i[t]$, it should be easy to see
that $\matrixm_i[t]\,v[t-1]$ is, in fact, identical to $v_i[t]$ obtained using (\ref{e_Z}). Thus, the above construction of $\matrixm_i[t]$ results in the contribution of the faulty nodes in $N_i^*[t]$ to (\ref{e_Z}) being replaced by an equivalent contribution from fault-free nodes in $\scriptl_g$ and $\scripts_g$.

\paragraph{Properties of $\matrixm_i[t]$:}

First, we show that $\matrixm[t]$ is row stochastic. Observe that all
the elements of $\matrixm_i[t]$ are non-negative.
Also, all the elements of $\matrixm_i[t]$ above add up to
\[
|N_i^*[t] \cap (\scriptv - F)|~a_i + |N_i^*[t] \cap F|~a_i = |N_i^*[t]|~a_i = 1
\]
because $a_i = 1/|N_i^*[t]|$ as defined in Algorithm 1.
Thus, $\matrixm_i[t]$ is a stochastic row vector.

Recall that from the above discussion, for $k\in N_i^*[t]$,
one of $S_k$ and $L_k$ must be $\geq 1/2$.
Without loss of generality, assume that $S_s \geq 1/2$ for some $s \in N_i^*[t] \cap F$.
Consequently, for each node $j \in \scripts_g$, $\matrixm_{ij}[t] \geq \frac{a_i}{|\scripts_g|} S_s \geq \frac{a_i}{2|\scripts_g|}$. Also,
 for each fault-free node $j$ in $N_i^*[t]$,
 $\matrixm_{ij}[t] = a_i$.
Thus, if $\beta$ is chosen such that
\begin{equation}
\label{eq:beta_caseI}
0 < \beta \leq \frac{a_i}{2|\scripts_g|}
\end{equation}
and $F_x(i)$ is defined to be equal to $\scriptl$, then the condition in the
lemma holds for node $i$. That is,
$\matrixm_{ij}[t] \geq \beta$ for $j \in \{i\} \cup ((\scriptv_{F} - F_x(i)) \cap N_i^-)$. 

\paragraph{All Cases Together:} Using similar constructions in other cases
as well (presented in Appendix \ref{app_s_construction}) and a suitable choice of $\beta$ (presented in Appendix \ref{a_s_together}), we can obtain a row stochastic matrix $\matrixm[t]$, and for each $i\in \scriptv-F$ identify a
feasible fault set $F_x(i)$, such that $\matrixm_{ij}[t] \geq \beta$ for all $j \in \{i\} \cup ((\scriptv_{F} - F_x(i)) \cap N_i^-)$.
Thus, Lemma \ref{lemma:tm2cm} can be proved correct. 

\end{proof}

\subsection{Validity and Convergence of Algorithm 1}
\label{s_proof}


The rest of the proof structure is derived from our previous work
wherein we proved the correctness of an IABC algorithm for the $f$-total fault
model \cite{vaidya_matrix_IABC}.
Let $R_{F}$ denote the set of all the reduced graphs of $G(\scriptv, \scripte)$
corresponding to a feasible fault set $F$.
Let $\tau = |R_{F}|$. $\tau$ depends on $F$ and the underlying network, and is finite.

In this discussion, let us denote a reduced graph by an italic upper case letter, and the corresponding ``connectivity matrix'' (defined below) using the same letter in boldface upper case. Thus, $\matrixh$ denotes the connectivity matrix for graph $\graphh \in R_{F}$.

Non-zero elements of connectivity matrix $\matrixh$ are defined as follows: (i) for $1 \leq i,j \leq n-\psi$, $\matrixh_{ij} = 1$ if and only if $(j, i) \in \graphh$, and (ii) $\matrixh_{ii} = 1$ for $1 \leq i \leq n-\psi$. That is, non-zero elements of row $\matrixh_i$ correspond to the incoming links at node $i$, and the self-loop at node $i$. Thus, the connectivity matrix for any reduced graph in $R_{F}$ has a non-zero diagonal.

Based on the {\em sufficient condition} stated at the start of Section \ref{s_sufficiency} and Lemma \ref{lemma:tm2cm}, we can show the following key lemmas.
 The proofs are presented in Appendix \ref{a_s_non-zero} and \ref{a_s_cm}.

\begin{lemma}
\label{lemma:non-zero}
For any $\graphh \in R_{F}, {\normalfont\bf \matrixh^{n-\psi}}$ has at least one non-zero column.
\end{lemma}

\begin{lemma}
\label{lemma:cm}
For any $t \geq 1$, there exists a graph $\graphh \in R_{F}$ such
 that $\beta {\normalfont\bf\matrixh \leq \matrixm}[t]$.
\end{lemma}


\begin{theorem}
\label{thm:sufficiency}
Suppose that $G(\scriptv, \scripte)$ satisfies the sufficient condition stated above. Algorithm 1 satisfies both the validity and convergence conditions.
\end{theorem}

\begin{proof}
A complete proof is presented in Appendix \ref{a_sufficiency}. By repeated application of (\ref{matrix:alg1}), we can represent the {\em Update} step of Algorithm 1 at the $t$-th iterations ($t \geq 1$) as:
\begin{eqnarray}
v[t] & = & \left(\,\Pi_{i=1}^t \matrixm[i]\,\right)\, v[0]
\label{e_v_t}
\end{eqnarray}
where $\matrixm[i]$ is constructed as described above.
When presenting matrix products, for convenience of presentation, we adopt
the following convention: for $a<b$, $\Pi_{i=a}^b \bfA[i]$ denotes
the ``backward'' product $\bfA[b]\bfA[b-1]\cdots\bfA[a]$.
Thus, $\Pi_{i=1}^t \bfM[i]$ in (\ref{e_v_t}) above represents
$\bfM[t]\bfM[t-1]\cdots\bfM[1]$.

Since $\matrixm[i]$ is row stochastic, then from (\ref{matrix:alg1}),
 it follows that Algorithm 1 satisfies the validity condition. Based on Lemmas \ref{lemma:non-zero} and \ref{lemma:cm}, we can also show that the rows of $\Pi_{i=1}^t \matrixm[i]$ become identical in the limit (as elaborated in
Appendix \ref{a_sufficiency}). This observation and (\ref{e_v_t}) together
imply that the states of the fault-free nodes satisfy the convergence condition
too.
\end{proof}

\section{Conclusions}
\label{s_conclusions}

This paper considers a {\em generalized} fault model, which can be used to specify more complex failure patterns, such as correlated failures or non-uniform node reliabilities. Under this fault model, we prove a {\em tight} necessary and sufficient condition for the existence of synchronous iterative approximate Byzantine consensus algorithms in arbitrary directed graphs. The analysis of consensus under the generalized fault model sheds new light on how the fault model affects algorithm design.

\newpage

\appendix

\setlength {\parskip}{6pt}  

\centerline{\Large\bf APPENDIX}


\section{Necessity Proof in Section \ref{s_necessity}}
\label{app_s_necessity}

Now, we present the proof for Theorem \ref{thm:nc}. The proof is by contradiction. Let us assume that a correct IABC algorithm exists, and
for some feasible fault set $F$, and 
feasible sets $F_x(i)$ for each $i\in\scriptv-F$, the resulting
reduced graph contains two source components.

Let $L$ and $R$ denote the
nodes in the two source components, respectively. Thus, $L$ and $R$ are disjoint and non-empty.
Let $C = (\scriptv - F - L -R)$ be the remaining nodes in the reduced graph.
$C$ may or may not be non-empty.
Let us now assume that the nodes in $F$ (if non-empty) are all faulty, and all the nodes in $L$,
$R$, and $C$ (if non-empty) are fault-free.

Consider the case when (i) each node in $L$ has initial input $m$ , (ii) each node in $R$ has initial input $M$, such that $M > m$, and (iii) each node in $C$ (if non-empty) has an input in the interval $[m, M]$.

In the {\em Transmit step} of iteration 1 of the IABC algorithm, suppose that the faulty nodes in $F$ (if non-empty) send $m^- < m$ on outgoing links to nodes in $L$, send $M^+ > M$ on outgoing links to nodes in $R$, and send some arbitrary value in interval $[m, M]$ on outgoing links to nodes in $C$ (if non-empty). This behavior is possible since nodes in $F$ are Byzantine faulty. Note that $m^- < m < M < M^+$. Each fault-free node $k \in \scriptv - F$ sends to nodes in $N_k^+$ value $v_k[0]$ in iteration 1.

\newcommand{\Ni}{N_i^- \cap (C \cup R)}

Consider any node $i \in L$. Since $L$ is a source component
in the reduced graph, 
it must be true that $\Ni \subseteq N_i^- \cap F_x(i)\cap\scriptv_F$.\footnote{Explanation:
In the reduced graph, there are no incoming links at $i$
from nodes in $\Ni$. Thus, any incoming links in $\scripte$ from  the nodes
in $\Ni$ must have been removed when constructing $\scripte_F$ for the reduced graph.
Recall that when constructing $\scripte_F$, incoming links from nodes
in $N_i^-\cap F_x(i)\cap \scriptv_F$ are removed.
It should be noted that the algorithm is performed using the links in $\scripte$, not
the reduced graph. Thus, in the {\em Transmit step}, all links in $\scripte$ are used.
}

Now, node $i$ receives $m^-$ from the nodes in $N_i^- \cap F$, and values in $[m,M]$ from the nodes in $\Ni$, and $m$ from the nodes in $\{i\}\cup(N_i^- \cap L)$. Figure \ref{f_necessity} illustrates the behavior of faulty nodes in $F$ and the value received by node $i$.

\begin{figure}
\centering
\includegraphics[width=0.7\textwidth]{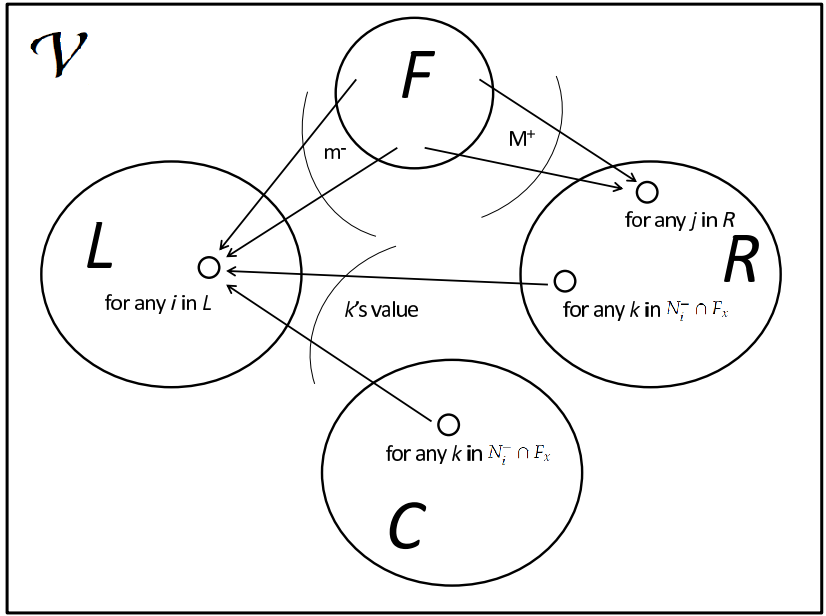}
\caption{Illustration of the behavior of faulty nodes in $F$ and the value received at node $i$.}
\label{f_necessity}
\end{figure}

Consider the following two cases:

\begin{itemize}
\item {\bf $N_i^- \cap F$ and $\Ni$ are both non-empty:} In this case, $(N_i^- \cap F) \subseteq F$ and $\Ni = N_i^-\cap F_x(i)\cap \scriptv_F \subseteq F_x(i)$. From node $i$'s perspective, consider two possible scenarios: (a) nodes in $N_i^- \cap F$ are all faulty, and the other nodes are fault-free, and (b) nodes in $\Ni = N_i^- \cap F_x(i)\cap \scriptv_F$ are all faulty, and the other nodes are fault-free.
Note that, since $F_x(i)$ is a feasible fault set, $N_i^- \cap F_x(i)\cap \scriptv_F$ is also
a feasible fault set.
Similarly,
since $F$ is a feasible fault set, $N_i^- \cap F$ is also a feasible fault set.

In scenario (a), from node $i$'s perspective, the fault-free nodes have sent values in interval $[m, M]$, whereas the faulty incoming neighbors, i.e., nodes in $N_i^- \cap F$, have sent value $m^-$. According to the validity condition, $v_i[1] \geq m$. On the other hand, in scenario (b), the fault-free incoming neighbors have sent values $m^-$ and $m$, where $m^- < m$; so $v_i[1] \leq m$, according to the validity condition. Since node $i$ does not know whether the correct scenario is (a) or (b), it must update its state to satisfy the validity condition in both cases. Thus, it follows that $v_i[1] = m$.

\item {\bf At most one of $N_i^- \cap F$ and $\Ni$ is non-empty:}
Recall that $N_i^-\cap F$ and $\Ni = N_i^- \cap F_x(i)\cap \scriptv_F$ are both feasible fault sets.
 Since at least one of these two sets is empty, their union, i.e., $(N_i^- \cap F)\cup (\Ni)$, is also a feasible fault set.

Then, from node $i$'s perspective, it is possible that all the nodes in $(N_i^- \cap F)\cup (\Ni)$ are faulty, and the rest of the nodes are fault-free. In this situation, the values sent to node $i$ by the fault-free nodes (which are all in $\{i\}\cup(N_i^- \cap L)$) are all $m$, and therefore, $v_i[1]$ must be set to $m$ as per the validity condition.

\end{itemize}
Hence, $v_i[1] = m$ for each node $i \in L$. Similarly, we can show that $v_j[1] = M$ for each node $j \in R$.

Now consider the nodes in set $C$ (if non-empty). All the values received by the nodes in $C$ are in $[m, M]$, therefore, their new state must also remain in $[m,M]$, as per the validity condition.

The above discussion implies that, at the end of iteration 1, the following conditions hold true:
(i) state of each node in $L$ is $m$ , (ii) state of each node in $R$ is $M$, and (iii) state of each node in $C$ (if non-empty) is in the interval $[m, M]$. These conditions are identical to the initial conditions listed previously. Then, by a repeated application of the above argument (proof by induction), it follows that for any $t \geq 0, v_i[t] = m$ for all nodes $i \in L$, $v_j[t] = M$ for all nodes $j \in R$ and $v_k[t] \in [m,M]$ for all nodes $k \in C$.

Since $L$ and $R$ both contain fault-free nodes, and $m\neq M$, the {\em convergence} requirement is not satisfied. This is a contradiction to the assumption that a correct iterative algorithm exists in $G(\scriptv, \scripte)$.

\section{Construction for other Cases in Section \ref{s_construction}}
\label{app_s_construction}

When discussing Case I in Section \ref{s_construction}, we deferred discussion
of the other cases. We present the construction for the rest of the cases here. There are six cases in total:

\begin{itemize}
\item Case I: $\scripts_g \neq \Phi, \scriptl_g \neq \Phi$, and $N_i^*[t] \cap F \neq \Phi$.
\item Case II: $\scripts_g \neq \Phi, \scriptl_g \neq \Phi$, and $N_i^*[t] \cap F = \Phi$.
\item Case III: $\scripts_g = \Phi, \scriptl_g \neq \Phi$, and $N_i^*[t] \cap F \neq \Phi$.
\item Case IV: $\scripts_g \neq \Phi, \scriptl_g = \Phi$, and $N_i^*[t] \cap F \neq \Phi$.
\item Case V: $\scripts_g = \Phi, \scriptl_g = \Phi$, and $N_i^*[t] \cap F \neq \Phi$.
\item Case VI: at most one of $\scripts_g$ and $\scriptl_g$ is non-empty,
 and  $N_i^*[t] \cap F = \Phi$.
\end{itemize}

Note that the choice of $f_1$ and $f_2$ in Algorithm 1 ensures that the value from node $i$ itself is never dropped from $r_i[t]$; therefore, $i \in N_i^*[t]$, and $N_i^*[t]$ is always non-empty. 

\subsection{Case II}

Now, we consider the case when $\scripts_g \neq \Phi, \scriptl_g \neq \Phi$, and $N_i^*[t] \cap F = \Phi$. That is, when each of $\scripts$ and $\scriptl$ contains at least one fault-free node, and $N_i^*[t]$ contains only fault-free node(s). In fact, the analysis of Case II is very similar to the analysis presented in Section \ref{s_construction} for Case I when $N_i^*[t]$ does contain a faulty node. 

We now discuss how the analysis of Case I can be applied to Case II. Rewrite (\ref{e_Z}) as follows:

\begin{eqnarray}
v_i[t] & = & \frac{a_i}{2} v_i[t-1]  + \frac{a_i}{2} v_i[t-1]
	+ \sum_{j\in N_i^*[t] - \{i\}} a_iw_j \\
& = & a_iw_z + a_i w_i
	+ \sum_{j\in N_i^*[t] - \{i\}} a_iw_j 
\end{eqnarray}

In the above equation, $z$ is to be viewed as a ``virtual'' incoming
neighbor of node $i$, which has sent value $w_z=\frac{v_i[t-1]}{2}$
to node $i$ in iteration $t$.
With the above rewriting of state update,
the value received by node $i$ from itself should
be viewed as $w_i=\frac{v_i[t-1]}{2}$ instead of $v_i[t-1]$. 
With this transformation, Case II now becomes identical
to Case I, with virtual node $z$ being treated
as an incoming neighbor of node $i$.

In essence, a part of node $i$'s contribution (half, to be precise) is now replaced by equivalent contribution by nodes in $\scriptl_g$ and $\scripts_g$. 
We now define elements $\matrixm_{ij}[t]$ of row $\matrixm_i[t]$:

\begin{itemize}
\item For $j = i$: $\matrixm_{ij}[t] = \frac{a_i}{2}$. This is obtained by observing in (\ref{e_Z}) that node $i$'s contribution to the new state $v_i[t]$ is $a_i\frac{v_i[t-1]}{2}$.

\item For $j \in N_i^*[t] - \{i\}$ : In this case, $j$ is a fault-free incoming neighbor of $i$. For each such $j$, define $\matrixm_{ij}[t] = a_i$. This is obtained by observing in (\ref{e_Z}) that the contribution of node j to the new state $v_i[t]$ is $a_iw_j = a_i v_j[t - 1]$.

\item For $j\in \scripts_g \cup \scriptl_g$ : In this case, $j$ is a fault-free node in $\scripts$ or  $\scriptl$. 

For each $j \in \scripts_g$,
\[
\matrixm_{ij}[t] ~=~ \frac{a_i}{2} \, \frac{S_z}{|\scripts_g|}
\]
and for each node $j \in \scriptl_g$,
\[
\matrixm_{ij}[t] ~=~ \frac{a_i}{2} \, \frac{L_z}{|\scriptl_g|}
\]

where $S_z$ and $L_z$ are chosen such that $S_z + L_z = 1$ and $w_z = \frac{v_i[t-1]}{2} = \frac{S_z}{2} m_{\scripts} + \frac{L_z}{2} m_{\scriptl}$. Note that such $S_z$ and $L_z$ exist because by definition of $\scripts_g$ and $\scriptl_g$, $v_i[t-1] \geq w_j,~\forall j \in S_g$ and $v_i[t-1] \leq w_j,~\forall j \in L_g$. Then the two expressions above are obtained by  replacing the contribution of the virtual node $z$ by an equivalent contribution by the nodes in $\scripts_g$ and $\scriptl_g$, respectively.

\item For $j\in (\scriptv - F) - (N_i^*[t] \cup \scripts_g \cup \scriptl_g)$ :
 These fault-free nodes have not yet been considered above.
For each such node $j$, define $\matrixm_{ij}[t] = 0$.
\end{itemize}

By argument similar to that in Section \ref{s_construction}, $\matrixm[t]$ is row stochastic. Without loss of generality, suppose that $S_z \geq 1/2$. Then for each node $j \in \scripts_g$, $\matrixm_{ij}[t] = \frac{a_i}{2|\scripts_g|}S_z \geq \frac{a_i}{4|\scripts_g|}$. Also, for fault-free node $j$ in $N_i^*[t]-\{i\}$, $\matrixm_{ij}[t] = a_i$, and $\matrixm_{ii}[t] = \frac{a_i}{2}$. Recall that by definition, $|\scripts_g| \geq 1$. Hence, if $\beta$ is chosen such that 

\begin{equation}
\label{eq:beta-caseII}
0 < \beta \leq \frac{a_i}{4|\scripts_g|}
\end{equation}
and $F_x(i)$ is defined to be equal to $\scriptl$, then the condition in the Lemma \ref{lemma:tm2cm} holds for node $i$. That is,
$\matrixm_{ij}[t] \geq \beta$ for $j \in \{i\} \cup (\scriptv_{F} - F_x(i)) \cap N_i^-$. 

\subsection{Cases III and IV}

Now, we describe the construction of Case III. The construction
for Case IV is very similar, and thus, is omitted here. 

In Case III, $\scripts_g = \Phi, \scriptl_g \neq \Phi$,
and $N_i^*[t]\cap F\neq \Phi$.
Thus, $\scripts$ does not contain any fault-free nodes (hence $\scripts_g$
is empty).
This may be due to one of the following two reasons: (i) the set
$\scripts$ is non-empty, but all the nodes in $\scripts$ are faulty,
or (ii) set $\scripts$ is empty.


Assume that $l\in\scriptl$ is a fault-free node, and that all the
nodes in $\scripts$ are faulty (i.e., $\scripts_g=\Phi$) or that $\scripts$ is empty (i.e., $f_1 = 0$).
In this case, observe that node $D(f_1+1)$ must be fault-free (otherwise, $f_1$
cannot be the largest value as defined in Algorithm 1).
Now, consider any node $k \in N_i^*[t]$. Similar to the argument in Case I,
we can find weights $S_k \geq 0$ and $L_k \geq 0$ such that
\[
S_k + L_k = 1
\]
and
\begin{equation}
\label{eq:caseII}
w_k = S_k~v_{D(f_1+1)}[t-1] + L_k~v_{l}[t-1]
\end{equation}
We now define $\matrixm_{ij}[t]$ for all fault-free $j$.

\begin{itemize}
\item For $j\in (N_i^*[t] - \{D(f_1+1)\})\cap (\scriptv-F)$. That is,
$j$ is a fault-free node in $N_i^*[t]$ with the exception of $D(f_1+1)$.

For each such $j$, define $\matrixm_{ij}[t] = a_i$. This is obtained by observing in (\ref{e_Z}) that the contribution of node $j$ to the new state $v_i[t]$ is $a_iw_j = a_i~v_j[t-1]$.

The elements of $\matrixm_i[t]$ defined here (including the case of $j=i$) add up to $$(|N_i^*[t] \cap (\scriptv - F)|-1)~a_i.$$

\item For nodes $D(f_1+1)$ and $l$: Define

\[
\matrixm_{iD(f_1+1)}[t] = a_i + \sum_{k \in N_i^*[t] \cap F} a_i~S_k
\]

and

\[
\matrixm_{il}[t] = \sum_{k \in N_i^*[t] \cap F} a_i~L_k
\]

Similar to Case I presented in Section \ref{s_construction}, these two expressions are obtained by summing up the contribution over the faulty nodes in $N_i^*[t]$, and replacing the sum by an equivalent contribution
 by the nodes $D(f_1+1)$ and $l$, respectively, according to (\ref{eq:caseII}).

The above elements of $\matrixm_i[t]$ add up to $$a_i~\left(1+\sum_{k \in N_i^*[t] \cap F}(S_k + L_k)\right) = (1+|N_i^*[t] \cap F|)~a_i.$$

\item For $j\in (\scriptv - F) - (N_i^*[t] \cup \{l\})$:
These fault-free nodes have not yet been considered above. For each such
$j$, define $M_{ij}[t] = 0$.
\end{itemize}
Similar to Case I, in Case III as well,
it should be easy to see that
\[
\matrixm_i[t]\,v[t-1]
\]
is identical to $v_i[t]$ obtained using (\ref{e_Z}).

\paragraph{Properties of $\matrixm_i[t]$:}
All the elements of $\matrixm_i[t]$ are non-negative.
The elements of $\matrixm_i[t]$ defined in Case II add up to
\[
(|N_i^*[t] \cap (\scriptv - F)|-1)~a_i + (1+|N_i^*[t] \cap F|)~a_i = |N_i^*[t]|~a_i = 1
\]
Thus, $\matrixm_i[t]$ is a stochastic row vector.

In Case III, recall that for any fault-free node $j$ in $N_i^*[t]$
(including $j= D(f_1+1)$ and $j=i$), $\matrixm_{ij}[t] \geq a_i$.
Thus, if $\beta$ is chosen such that
\begin{equation}
\label{eq:beta-caseIII}
0 < \beta \leq a_i
\end{equation}
and $F_x(i)$ is defined to be equal to $\scriptl$,
then the condition in the
Lemma \ref{lemma:tm2cm} holds for node $i$. 

\subsection{Case V}

Consider Case V, where $N_i^*[t] \cap F \neq \Phi$, and $\scripts_g = \scriptl_g = \Phi$. In this case, it should be easy to see that $N_i^*[t]$ contains at least 3 nodes. In particular, $D_{f_1+1}$ must be fault-free (otherwise, $f_1$ cannot be maximum possible), $D_{|N_i^-|-f_2+1}$ must be fault-free (otherwise, $f_2$ cannot be maximum possible), and there is a faulty node in $N_i^*[t]$.

Now this case can be handled similar to Case III analyzed above. In particular, entries in $\matrixm_i[t]$ are defined similarly with $l$ being defined equal to $D_{N_i^--f_2+1}$.  Also, define $F_x(i) = \Phi$. 

Hence, it is easy to see that the properties of $\matrixm_i[t]$ are identical to Case III presented above.

\subsection{Case VI}

Here, we consider the case when at most one of $\scripts$ and $\scriptl$ contains a fault-free node and $N_i^*[t] \cap F = \Phi$. Without loss of generality, suppose that $\scripts$ contains only faulty nodes, and $\scriptl$ may contain a fault-free node.

In this case, define $\matrixm_{ij}[t]=a_i$ for $j\in N_i^*[t]$;
define $\matrixm_{ij}=0$ for all other fault-free nodes $j$.
Also, define $F_x(i)=\scriptl$.

The properties of $\matrixm_i[t]$ thus defined are identical to Case III above.

\section{Putting Cases Together}
\label{a_s_together}

Now, let us consider Cases I-VI together. From the definition of $a_i$ in Algorithm 1, observe that $a_i \geq \frac{1}{|N_i^-|+1}$ (because $f_1,f_2\geq 0$). Let us define
\[ \alpha = \min_{i\in\scriptv} \frac{1}{|N_i^-|+1}\]
Moreover, observe that $|\scripts_g| \leq n$ and $|\scriptl_g| \leq n$. Then define $\beta$ as 
\begin{equation}
\label{eq:beta}
\beta = \frac{\alpha}{4n}
\end{equation}
This definition satisfies constraints on $\beta$ in Cases I through VI (conditions (\ref{eq:beta_caseI}), (\ref{eq:beta-caseII}) and (\ref{eq:beta-caseIII})). Thus, Lemma \ref{lemma:tm2cm} holds for all six cases with this choice of $\beta$.

\section{Proof of Lemma \ref{lemma:non-zero} in Section \ref{s_proof}}
\label{a_s_non-zero}

Here, we present the proof of the first key lemma used in the sufficiency proof. \\

\noindent {\bf Lemma \ref{lemma:non-zero}  }
{\em For any $\graphh \in R_{F}, {\normalfont\bf \matrixh^{n-\psi}}$ has at least one non-zero column.}

\begin{proof}
$G(\scriptv,\scripte)$ satisfies the {\em sufficient condition}
stated at the start of Section \ref{s_sufficiency}. Therefore,
there exists at least one non-faulty node $k$ in the reduced graph $\graphh$ that has
directed paths to all the nodes in $\graphh$ (consisting of the edges in $\graphh$).
 Since the length of the path from $k$ to any other node in $\graphh$ is at most $n-\psi-1$, the $k$-th column of matrix $\matrixh^{n-\psi}$ will be non-zero.\footnote{That is, all the elements of the column will be non-zero. Also, such a non-zero column will exist in $\matrixh^{n-\psi-1}$, too. We use the loose bound of $n-\psi$ to simplify the presentation.}
\end{proof}

\section{Proof of Lemma \ref{lemma:cm} in Section \ref{s_proof}}
\label{a_s_cm}

Here, we present the proof of the second key lemma used in the sufficiency proof. We start with two definitions:

\begin{definition}
For matrices $\textbf{A}$ and $\textbf{B}$ of identical size, and a scalar $\gamma$, $\gamma \textbf{B} \leq \textbf{A}$ provided that $\gamma \textbf{B}_{ij} \leq \textbf{A}_{ij}$ for all $i, j$.
\end{definition}
We want to prove the following lemma.

~

\noindent {\bf Lemma \ref{lemma:cm}  }
{\em For any $t \geq 1$, there exists a graph $\graphh \in R_{F}$ such that $\beta {\normalfont\bf\matrixh \leq \matrixm}[t]$.}

\begin{proof}
Observe that the $i$-th row of the transition matrix $\matrixm[t]$ corresponds to
the state update (in Algorithm 1) performed at fault-free node $i$.
Recall from Lemma \ref{lemma:tm2cm} that $\matrixm_{ij}[t] \geq \beta$
for $j \in \{i\} \cup ((\scriptv_{F} - F_x(i)) \cap N_i^-)$, where
$F_x(i)$ is a feasible fault set.

Let us obtain a reduced graph $H$ by choosing $F_x(i)$ for each $i$ as defined
in Lemma \ref{lemma:tm2cm}.
Then from the definition of connectivity matrix $\matrixh$, Lemma \ref{lemma:cm} then follows.
\end{proof}

\section{Correctness of Algorithm 1}
\label{a_sufficiency}

When presenting matrix products, for convenience of presentation, we adopt
the following convention: for $a<b$, $\Pi_{i=a}^b \bfA[i]$ denotes
the ``backward'' product $\bfA[b]\bfA[b-1]\cdots\bfA[a]$.

The proof below is similar to a proof for the $f$-total fault model in our previous work \cite{vaidya_matrix_IABC}. It is included here for the
convenience of the referees.

\subsection{Matrix Preliminaries}

In the discussion below, we use boldface upper case letters to denote matrices,
rows of matrices, and their elements. For instance,
$\bfH$ denotes a matrix, $\bfH_i$ denotes the $i$-th row of
matrix $\bfH$, and $\bfH_{ij}$ denotes the element at the
intersection of the $i$-th row and the $j$-th column
of matrix $\bfH$.

\begin{definition}
\label{d_stochastic}
A vector is said to be {\em stochastic} if all the elements
of the vector are {\em non-negative}, and the elements add up to 1.
A matrix is said to be \underline{row} stochastic if each row of the matrix is a
stochastic vector. 
\end{definition}

For a row stochastic matrix $\bfA$,
 coefficients of ergodicity $\delta(\bfA)$ and $\lambda(\bfA)$ are defined as
follows \cite{Wolfowitz}:
\begin{eqnarray*}
\delta(\bfA) & = &   \max_j ~ \max_{i_1,i_2}~ | \bfA_{i_1\,j}-\bfA_{i_2\,j} | \label{e_zelta} \\
\lambda(\bfA) & = & 1 - \min_{i_1,i_2} \sum_j \min(\bfA_{i_1\,j} ~, \bfA_{i_2\,j}) \label{e_lambda}
\end{eqnarray*}
It is easy to show that  $0\leq \delta(\bfA) \leq 1$ and $0\leq \lambda(\bfA) \leq 1$, and that the rows
of $\bfA$ are all identical if and only if $\delta(\bfA)=0$. Also, $\lambda(\bfA) = 0$ if and only if $\delta(\bfA) = 0$.

The next result from \cite{Hajnal58} establishes a relation between the coefficient of ergodicity $\delta(\cdot)$ of a product of row stochastic matrices, and the coefficients of ergodicity $\lambda(\cdot)$ of the individual matrices defining the product. 

\begin{lemma}
\label{claim_zelta}
For any $p$ square row stochastic matrices $\bfQ(1),\bfQ(2),\dots \bfQ(p)$, 
\begin{eqnarray*}
\delta(\bfQ(p)\bfQ(p-1)\cdots \bfQ(1)) ~\leq ~
 \Pi_{i=1}^p ~ \lambda(\bfQ(i)).
\end{eqnarray*}
\end{lemma}

Lemma \ref{claim_zelta} is proved in \cite{Hajnal58}. It implies that
if, for all $i$, $\lambda(\bfQ(i))\leq 1-\gamma$ for some $\gamma$, where $0<\gamma\leq 1$, then $\delta(\bfQ(p)\bfQ(p-1)\cdots \bfQ(1))$ will approach zero as $p$ approaches $\infty$. 
We now define a {\em scrambling} matrix \cite{Hajnal58,Wolfowitz}. 

\begin{definition}
A row stochastic
 matrix $\bfH$ is said to be a {\em scrambling}\, matrix if $\lambda(\bfH)<1$.
\end{definition}

The following lemma follows easily from the above definition of $\lambda(\cdotp)$. 
\begin{lemma}
\label{l_lambda_bound}
If any column of a row stochastic matrix $\bfH$
contains only non-zero elements that are all lower bounded by some
constant $\gamma$, where $0<\gamma\leq 1$, then $\bfH$ is a scrambling matrix, and $\lambda(\bfH)\leq 1-\gamma$. 
\end{lemma}

\subsection{Correctness of Algorithm 1}

\begin{lemma}
\label{l_product_H}
For any $z\geq 1$,
in the product below of $\bfH[t]$ matrices for consecutive
$\tau(n-\psi)$ iterations, at least one column is non-zero. 
\[
\Pi_{t=z}^{z+\tau(n-\psi)-1} \, \bfH[t]
\]
\end{lemma}
\begin{proof}
Since the above product consists of $\tau(n-\psi)$ connectivity matrices
corresponding to graphs
in $R_\scriptf$,
at least one of the connectivity matrices
corresponding to the $\tau$ distinct graphs
in $R_\scriptf$, say matrix $\bfH_*$\,, will appear in the above
product at least $n-\psi$ times.

Now observe that: (i)
By Lemma \ref{lemma:non-zero}, $\bfH_*^{n-\psi}$ contains a non-zero
column, say the $k$-th column is non-zero,
and (ii) all the $\bfH[t]$ matrices in the product contain a non-zero diagonal.
These two observations together imply that the $k$-th column in the above product 
is non-zero.
\end{proof}

~

Let us now define a sequence of matrices $\bfQ(i)$, $i\geq 1$, such that
each of these matrices is a product of $\tau(n-\psi)$ of the
$\bfM[t]$ matrices. Specifically,
\begin{eqnarray}
\bfQ(i) &=& \Pi_{t=(i-1)\tau(n-\psi)+1}^{i\tau(n-\psi)} ~ \bfM[t]
\label{e_Q_i}
\end{eqnarray}
From (\ref{e_v_t}) and (\ref{e_Q_i})
observe that
\begin{eqnarray}
\bfv[k\tau(n-\psi)] & = & \left(\, \Pi_{i=1}^k ~ \bfQ(i) \,\right)~\bfv[0]
\end{eqnarray}

\begin{lemma}
\label{l_Q}
For $i\geq 1$, $\bfQ(i)$ is a scrambling row stochastic matrix,
and \[ \lambda(\bfQ(i))\leq 1-\beta^{\tau(n-\psi)}.\]
\end{lemma}
\begin{proof}

$\bfQ(i)$ is a product of row stochastic matrices ($\bfM[t]$); therefore,
$\bfQ(i)$ is row stochastic.
From Lemma \ref{lemma:cm}, for each $t\geq 1$,
\[
\beta \, \bfH[t] ~ \leq ~ \bfM[t]
\]
Therefore, 
\[
\beta^{\tau(n-\psi)} ~ \Pi_{t=(i-1)\tau(n-\psi)+1}^{i\tau(n-\psi)} ~ \bfH[t] ~ \leq 
~ \Pi_{t=(i-1)\tau(n-\psi)+1}^{i\tau(n-\psi)} ~ \bfM[t] ~ =
~ \bfQ(i)
\]
By using $z=(i-1)(n-\psi)+1$ in Lemma \ref{l_product_H},
we conclude that the matrix product on the left side
of the above inequality contains a non-zero column. Therefore, $\bfQ(i)$ on the
right side of the inequality also contains
a non-zero column.

Observe that $\tau(n-\psi)$ is finite, and hence, $\beta^{\tau(n-\psi)}$
is non-zero. Since the non-zero terms in $\bfH[t]$ matrices are all 1,
the non-zero elements in $\Pi_{t=(i-1)\tau(n-\psi)+1}^{i\tau(n-\psi)} \bfH[t]$
must each be $\geq$ 1. Therefore, there exists a non-zero column in $\bfQ(i)$
with all the elements in the column being $\geq \beta^{\tau(n-\psi)}$.
Therefore, by Lemma \ref{l_lambda_bound}, $\lambda(\bfQ(i))\leq 1-\beta^{\tau(n-\psi)}$, and $\bfQ(i)$ is a scrambling matrix.
\end{proof}

~

\noindent {\bf Theorem \ref{thm:sufficiency}  }
{\em Suppose that $G(\scriptv, \scripte)$ satisfies the sufficient condition stated above. Algorithm 1 satisfies both the validity and convergence conditions.}

\begin{proof}

Since $\bfv[t]=\bfM[t]\,v[t-1]$, and $\bfM[t]$ is a row stochastic matrix,
it follows that
Algorithm 1 satisfies the validity condition.

Using Lemma \ref{claim_zelta} and the definition of $\bfQ(i)$, 
and using the inequalities
$\lambda(\bfM[t])\leq 1$ and $\lambda(\bfQ(i))\leq (1-\beta^{\tau(n-\psi)})<1$, we get
\begin{eqnarray*}
\lim_{t\rightarrow \infty} \delta(\Pi_{i=1}^t \bfM[i])
~=~ \lim_{t\rightarrow\infty}  \delta\left(
\left(\Pi_{i=(\lfloor\frac{t}{\tau(n-\psi)}\rfloor)\tau(n-\psi)+1}^t \bfM[i]\right)
\left(\Pi_{i=1}^{\lfloor\frac{t}{\tau(n-\psi)}\rfloor} \bfQ(i)\right)\right)
 \\ 
~ \leq ~ \lim_{t\rightarrow\infty} \Pi_{i=1}^{\lfloor\frac{t}{\tau(n-\psi)}\rfloor} \lambda(\bfQ(i)) 
~ = ~ 0 
\end{eqnarray*}
Thus, the rows of $\Pi_{i=1}^t \bfM[i]$ become identical in the limit.
This observation, and the fact that $\bfv[t]=(\Pi_{i=1}^t \bfM[i])\bfv[0]$ together imply that
the states of the fault-free nodes satisfy the
convergence condition.
\end{proof}


\begin{thebibliography}{10}

\bibitem{AA_convergence_markov}
D.~P. Bertsekas and J.~N. Tsitsiklis.
\newblock {\em Parallel and Distributed Computation: Numerical Methods}.
\newblock Optimization and Neural Computation Series. Athena Scientific, 1997.

\bibitem{Vartika_radio_byzantine_2005}
V.~Bhandari and N.~H. Vaidya.
\newblock On reliable broadcast in a radio network.
\newblock In {\em Proceedings of the twenty-fourth annual ACM symposium on
  Principles of distributed computing}, PODC '05, pages 138--147, New York, NY,
  USA, 2005. ACM.

\bibitem{dag_decomposition}
S.~Dasgupta, C.~Papadimitriou, and U.~Vazirani.
\newblock {\em Algorithms}.
\newblock McGraw-Hill Higher Education, 2006.

\bibitem{AA_Dolev_1986}
D.~Dolev, N.~A. Lynch, S.~S. Pinter, E.~W. Stark, and W.~E. Weihl.
\newblock Reaching approximate agreement in the presence of faults.
\newblock {\em J. ACM}, 33:499--516, May 1986.

\bibitem{FLP_one_crash}
M.~J. Fischer, N.~A. Lynch, and M.~S. Paterson.
\newblock Impossibility of distributed consensus with one faulty process.
\newblock {\em J. ACM}, 32:374--382, April 1985.

\bibitem{Hajnal58}
J.~Hajnal.
\newblock Weak ergodicity in non-homogeneous markov chains.
\newblock In {\em Proceedings of the Cambridge Philosophical Society},
  volume~54, pages 233--246, 1958.

\bibitem{AA_PCN_Local}
R.~M. Kieckhafer and M.~H. Azadmanesh.
\newblock Low cost approximate agreement in partially connected networks.
\newblock {\em Journal of Computing and Information}, 3(1):53--85, 1993.

\bibitem{Koo_radio_byzantine}
C.-Y. Koo.
\newblock Broadcast in radio networks tolerating byzantine adversarial
  behavior.
\newblock In {\em Proceedings of the twenty-third annual ACM symposium on
  Principles of distributed computing}, PODC '04, pages 275--282, New York, NY,
  USA, 2004. ACM.

\bibitem{Nonuniform_failure_models}
P.~Kuznetsov.
\newblock Understanding non-uniform failure models.
\newblock {\em Bulletin of the European Association for Theoretical Computer
  Science (BEATCS)}, 106:53--77, 2012.

\bibitem{psl_BG_1982}
L.~Lamport, R.~Shostak, and M.~Pease.
\newblock The byzantine generals problem.
\newblock {\em ACM Trans. on Programming Languages and Systems}, 1982.

\bibitem{leblanc_HiCoNs}
H.~LeBlanc, H.~Zhang, S.~Sundaram, and X.~Koutsoukos.
\newblock Consensus of multi-agent networks in the presence of adversaries
  using only local information.
\newblock {\em HiCoNs}, 2012.

\bibitem{AA_nancy}
N.~A. Lynch.
\newblock {\em Distributed Algorithms}.
\newblock Morgan Kaufmann, 1996.

\bibitem{vaidya_matrix_IABC}
N.~H. Vaidya.
\newblock Matrix representation of iterative approximate byzantine consensus in
  directed graphs.
\newblock {\em CoRR}, Mar. 2012.

\bibitem{vaidya_PODC12}
N.~H. Vaidya, L.~Tseng, and G.~Liang.
\newblock Iterative approximate byzantine consensus in arbitrary directed
  graphs.
\newblock volume abs/1201.4183, 2012.

\bibitem{Wolfowitz}
J.~Wolfowitz.
\newblock Products of indecomposable, aperiodic, stochastic matrices.
\newblock In {\em Proceedings of the American Mathematical Society}, volume~14,
  pages 733--737, 1963.

\bibitem{IBA_broadcast_Sundaram}
H.~Zhang and S.~Sundaram.
\newblock Robustness of information diffusion algorithms to locally bounded
  adversaries.
\newblock {\em CoRR}, abs/1110.3843, 2011.

\end{thebibliography}
\end{document}